\documentclass[12pt]{amsart}%
\usepackage[english]{babel}
\usepackage{amssymb}
\usepackage{amsmath}
\usepackage{amsfonts}
\usepackage{verbatim}

\numberwithin{equation}{section}

\newtheorem{theorem}{Theorem}
\newtheorem{lemma}{Lemma}[section]
\newtheorem{prop}[lemma]{Proposition}
\newtheorem{cor}[lemma]{Corollary}
\newtheorem*{rmk*}{Remark}

\newcommand{\dist}{\mathrm{dist}}
\newcommand{\meas}{\mathrm{meas}}

\newcommand{\openrm}{\mathrm{(}}
\newcommand{\closerm}{\mathrm{)}}
\newcommand{\op}{H_{\alpha, \lambda, \theta}}

\newcommand{\specU}{S(\alpha, \lambda)}

\newcommand{\dimension}{\mathrm{dim_H}}

\def\R{\mathbb{R}}
\def\N{\mathbb{N}}
\def\Q{\mathbb{Q}}

\begin{document}

\title[Zero Hausdorff Dimension Spectrum]{Zero Hausdorff Dimension Spectrum for the Almost Mathieu Operator.}
\author[Yoram Last and Mira Shamis]{Yoram Last\textsuperscript{1} and Mira Shamis\textsuperscript{2}}
\footnotetext[1]{Institute of Mathematics, The Hebrew University,
91904 Jerusalem, Israel. \mbox{E-mail}: ylast@math.huji.ac.il}
\footnotetext[2]{Current address: Department of Mathematics, The
Weizmann Institute of Science, Rehovot 7610001, Israel. E-mail:
mira.shamis@weizmann.ac.il}

\begin{abstract}
We study the almost Mathieu operator at critical coupling. We prove that 
there exists a dense $G_\delta$ set of frequencies for which the
spectrum is of zero Hausdorff dimension.
\end{abstract}

\maketitle




\section{Introduction}
\noindent The almost Mathieu operator $\op$\ is the discrete
one-dimensional Schr\"{o}dinger operator, acting on
$\ell^2(\mathbb{Z})$\ and defined by
\begin{equation}\label{almostMathieu}
\begin{array}{ll}
\op = \Delta + V_{\alpha, \lambda, \theta}, & \textrm{\ }
(\Delta\psi)(n) = \psi(n + 1) + \psi(n - 1),
\end{array}
\end{equation}
where
\begin{equation*}
(V\psi)(n) = \lambda\cos(2\pi{\alpha}n + \theta)\psi(n),
\end{equation*}
and $\alpha, \lambda, \theta \in\mathbb{R}$.\ Since
$H_{\alpha,\lambda, \theta}$\ is invariant under
$\alpha\rightarrow{\alpha \pm 1}$,\ $\theta\rightarrow{\theta \pm
2\pi}$,\ we may always assume that $0 \leq \alpha \leq 1$,\ $0 \leq
\theta \leq 2\pi$. Also, a sign change of $\lambda$  is equivalent to translation of $\theta$\ by $\pi$,
therefore we assume that $\lambda \geq 0$.

The almost Mathieu operator is one of the most studied concrete
models of one-dimensional Schr\"{o}dinger operators, with a rich
spectral theory and important applications in physics (see, e.g., \cite{AJ,J,last4,last_review}).

If $\alpha = \frac{p}{q}\in\mathbb{Q}$, where we may assume that $p
< q$, $p$ and $q$ are relatively prime, the operator is periodic,
and if $\alpha\in\mathbb{R\setminus{Q}}$, the operator is almost
periodic (meaning that the closure of the set $\{ V(\cdot +
j)\}_{j=-\infty}^{+\infty}$ is compact in $\ell^\infty$.)

In the periodic case the spectrum is the union of $q$\ bands (closed intervals),
possibly touching at the edges (endpoints). As $\theta$\ varies, these
bands move and their length may change, but different bands never
overlap other than at endpoints. In fact, it is known \cite{choi-elliot-yui,vmch}
that these bands are disjoint except for some special values of $\theta$ when $q$\
is even. In these cases, the two central bands are touching at $0$. 

In the almost periodic case, the spectrum $\sigma(\alpha, \lambda,
\theta)$ does not depend on $\theta$. In this case, for any
$\alpha\in\mathbb{R}\setminus\mathbb{Q}$ and $\lambda\neq 0$, it is
a Cantor set (a nowhere dense closed set with no isolated
points) and has Lebesgue measure $|4 - 2|\lambda\|$.
This theorem is the result of work by
numerous researchers over several decades (see \cite{last_review}).
The final step was accomplished in the work of Avila--Jitomirskaya \cite{AJ}.

\vspace{2mm}
For the case $|\lambda| = 2$, known as the critical case, and for every irrational
$\alpha$, the spectrum is a Cantor set of zero Lebesgue measure and the
question of its fractal dimension naturally arises.
Our main result is the following theorem:

\begin{theorem}\label{mainTheorem}
There exists a dense $G_\delta$\ set of $\alpha$'s,\ for which
\begin{equation*}
\dimension(\sigma(\alpha, 2, \theta)) = 0 \quad \text{for every $\theta$,}
\end{equation*}
where $\sigma(\alpha, 2, \theta)$\ is the spectrum of
$H_{\alpha,2,\theta}$ and $\dimension(\cdot)$\ denotes Hausdorff
dimension.
\end{theorem}

\vspace{2mm} It was believed until the mid 1990's that 
the box-counting  dimension $\dim_B(\sigma(\alpha,2,\theta))$ 
of the spectrum of the almost Mathieu operator is equal to
$\frac{1}{2}$\ for $|\lambda| = 2$\ and almost every $\alpha$;
see, e.g., Tang and Kohmoto \cite{conjhalf1}, Bell and Stinchcombe
\cite{conjhalf2}, Geisel, Ketzmerick and Petschel \cite{conjhalf3},
for numerical and heuristic arguments supporting this conjecture.
However, in 1994 Wilkinson--Austin \cite{wilkonson_austin}
provided numerical evidence that $\dim_B(\sigma(\alpha, 2, \theta))
= 0.498$\ for $\alpha = \frac{1 + \sqrt{5}}{2}$ (the golden mean)
and thus conjectured that $\dim_B(\sigma(\alpha, 2,
\theta)) < \frac{1}{2}$ for every irrational $\alpha$. They also
gave numerical and analytic evidence that $\dim_B(\sigma(\alpha_n,
2, \theta)) \rightarrow 0$\ as $n \rightarrow \infty$\ for
irrationals $\alpha_n$\ of the form
\[
\alpha_n = \frac{1}{n + \frac{1}{n + \frac{1}{n + \ldots}}}~,
\]
giving additional indication that fractal properties of
the spectrum $\sigma(\alpha, 2,
\theta)$ depend sensitively on the properties of $\alpha$.

As for Hausdorff dimension, J.~Bellissard (private communication to the first
author, circ.\ 1995) conjectured that there should
exist some $\beta \in (0, 1/2]$ such that
$\dimension(\sigma(\alpha, 2, \theta)) = \beta$ for almost every
$\alpha$. The only rigorous result on the Hausdorff dimension 
that we are aware of is in \cite{last3}, were it is proved that for a 
dense $G_\delta$ set of irrational
$\alpha$'s, the Hausdorff dimension of the spectrum is $\leq \frac{1}{2}$. The relevant
$\alpha$'s are characterized by
\[
\liminf_{q \to \infty} \inf_p q^4 \left| \alpha - \frac{p}{q} \right| < \infty~.
\]

\smallskip
We mention an immediate corollary of Theorem~\ref{mainTheorem} pertaining to the
integrated density of states. Recall that the integrated density of states
$\mathcal{N}_{\alpha, \lambda}(E)$ is a cumulative distribution function of a measure
supported on $\sigma(H_{\alpha,\lambda,\theta})$. For irrational $\alpha$,
it can be defined via its Stieltjes transform by
\[
\int f(E) d\mathcal{N}_{\alpha, \lambda}(E) = \frac{1}{2\pi} \int_0^{2\pi} 
\langle f(H_{\alpha,\lambda,\theta}) \delta_0, \delta_0 \rangle \, d\theta~, \quad f \in C(\mathbb{R})~,
\]
where $\delta_0(n)$ is $1$ if $n=0$ and $0$ otherwise.
Recalling Frostman's lemma (see, e.g., \cite{Falconer}), we obtain:

\begin{cor} There exists a dense $G_\delta$ set of $\alpha$'s for which
the integrated density of states $\mathcal{N}_{\alpha, 2}(E)$ is not H\"older 
continuous of any order $\gamma > 0$.
\end{cor}
We note that the existence of irrational $\alpha$'s for which $\mathcal{N}_{\alpha,2}$ is not H\"older 
continuous of any order $\gamma > 0$ also follows from the work of Helffer--Sj\"ostrand \cite{HS}; see 
Bourgain \cite[Ch.\ 8]{Bourg}.

\medskip A key step in the proof of Theorem~\ref{mainTheorem} is the following theorem.
Fix $q\in\mathbb{Z^+}$\ and consider the almost Mathieu
operator $H_{\frac{p}{q}, 2, \theta}$,\ where $p < q$\ are
relatively prime. Let 
\[ S_-\left(\frac{p}{q},
2\right) \equiv \bigcap_{\theta}\sigma\left(\frac{p}{q}, 2, \theta\right)~, \quad S\left(\frac{p}{q},
2\right) \equiv \bigcup_{\theta}\sigma\left(\frac{p}{q}, 2, \theta\right)~.\]
 As discussed below,
$S_-(p/q, 2)$ consists of $q$ points; let
\begin{equation}\label{eqdefdist} \begin{split}
\Delta_{\frac{p}{q}, 2}(E) &\equiv \prod_{E' \in S_-(p/q, 2)} (E - E')~,\\
\dist\left(E, S_{-}\left(\frac{p}{q},
2\right)\right)&\equiv{\inf_{E^\prime\in{S_{-}(\frac{p}{q}, 2)}}|E - E^\prime|}~.
\end{split}\end{equation}
Denote
\[ \begin{split}
J_\delta^{(1)}&\equiv \{E: |\Delta_{\frac{p}{q}, 2}(E)| > \delta > 0\}~,\\
{J}_\delta^{(2)}&\equiv \left\{E: \dist\left(E, S_{-}\left(\frac{p}{q}, 2\right)\right) >
 \delta > 0\right\}~.
 \end{split}\]
%

\begin{theorem}\label{lebesgueMeasure}
For any $p/q$ and $0 < \delta$ there exists $\eta = \eta(p,q,\delta)>0$ 
such that for $|\frac{\widetilde{p}}{\widetilde{q}} -
\frac{p}{q}| < \eta(p,q,\delta)$
\begin{equation*}
\meas\left(S\left(\frac{\widetilde{p}}{\widetilde{q}}, \,2\right) \cap J_\delta\right) < \
\frac{C(p,q)}{\delta} \exp \left\{-\frac{\delta\widetilde{q}}{C(p,q)}\right\}~,
\end{equation*}
where $C(p, q)$\ depends only on $p$ and $q$,
 $\meas(\cdot)$\ denotes the Lebesgue measure, and $J_\delta$
 denotes either $J_\delta^{(1)}$ or $J_\delta^{(2)}$. For $J_\delta^{(1)}$,
 one may choose $\eta(p,q,\delta) = C^{-q} \delta^2$ and $C(p, q) = Cq$.
\end{theorem}

This theorem  formalizes the following
observation inspired by 
the drawing of Hofstadter \cite{Hofstadter}. Consider some fixed
rational $\frac{p}{q}$,\ where $p$\ and $q$\ are relatively prime,
and consider a sequence $\frac{p_n}{q_n}\to \frac{p}{q}$. Then the spectrum naturally splits into
regions that correspond to the bands of the spectrum for
$\frac{p}{q}$; the bands that correspond to high order rationals
become exponentially small away from the center of the
corresponding region. 

\medskip
An important technical step in the proof of Theorem~\ref{lebesgueMeasure}
is the following result, which allows  to study the behavior of
solutions in certain intervals of energies and to obtain their
exponential growth for energies outside of the spectrum of the
operator. We denote by $\operatorname{Tr} T$ and $\det T$ the trace
and determinant (respectively) of a square matrix $T$.

\begin{theorem}\label{matrixTheorem}
Fix $0 < \beta < 1$. Suppose that $T_n \in SL_2(\mathbb{C})$, $1 \leq n \leq N$,
are matrices with eigenvalues $e^{\pm(\gamma_n+ i \zeta_n)}$, $\gamma_n > 0$,
and that the corresponding eigenvectors ${\phi}^{\pm}_n$ can be chosen
so that $\|{\phi}^{\pm}_n\| = 1$ and 
\[ \frac{4 \max(\|\phi^{+}_n- \phi^{+}_{n+1}\|,
\|\phi^{-}_n -
\phi^{-}_{n+1}\|)}{|\det (\phi_{n+1}^+, \phi_{n+1}^-)|} <
\beta\gamma_n{e^{-\gamma_n}}~.\]
Then the matrix product $\Phi_N = T_N\ldots T_1$ satisfies
\[
(1 - \beta)e^{(1 - \beta){\sum_{n=1}^N {\gamma _n}}}\leq\|\Phi_N
\phi^{+}_1\|\leq (1 + \beta)e^{(1 + \beta){\sum_{n=1}^N {\gamma
_n}}}.
\]
\end{theorem}

The rest of the paper is organized as follows. In
Section~\ref{s:am.prel} we describe some preliminaries and
previously obtained results. In Section~\ref{sec:matrixTheorem} we
prove Theorem~\ref{matrixTheorem} and in
Section~\ref{sec:lebesgueMeasure} we prove Theorem~\ref{lebesgueMeasure}.
Finally, in Section~\ref{sec:mainTheorem},
we prove Theorem~\ref{mainTheorem}.
One of the steps in the proof of Theorem~\ref{lebesgueMeasure} is inspired
by the work of Surace \cite{surace}, whereas the deduction of Theorem~\ref{mainTheorem}
makes use of several arguments from \cite{last3}.

\medskip
\paragraph{\bf Acknowledgements} YL is grateful to F.~Klopp and A.~Fedotov for helpful
discussions during preliminary stages of this project. MS is grateful to Qi Zhou for
helpful comments on the text. The authors are also very grateful to the referee of this paper
for suggesting several corrections and improvements.
This research was supported in part by The Israel Science Foundation (Grants No.\ 1169/06 and 1105/10)
and by Grants No.\ 2010348 and 2014337 from the United States--Israel Binational Science Foundation (BSF),
Jerusalem, Israel.

\section{Preliminaries}\label{s:am.prel}

\paragraph{\bf Periodic operators}
We start with some general facts about periodic operators of period
$\widetilde{q}$. Namely, we start with an operator $H$\ that acts on $\ell^2(\mathbb{Z})$, of the form 
\[
H = \Delta + V, \ 
(\Delta\psi)(n) = \psi(n + 1) + \psi(n - 1),\
\]
where $V$\ is periodic sequence of period $\widetilde{q}$.
Consider the associated eigenvalue equation
\begin{equation}\label{eigenvalue_eq}
\psi(n + 1) + \psi(n - 1) + V(n)\psi(n) = E\psi(n),
\end{equation}
where $E\in\mathbb{R}$ and $\psi$ is a two-sided sequence,
$\psi:{\mathbb Z}\rightarrow{\mathbb C}$.\

It is known that for such an operator the Lyapunov exponent exists for every $E\in\R$ and
it is determined by the one period transfer matrix
$\Phi_{\widetilde{q}}(E)$ as follows:
\begin{equation*}
\gamma(E) \equiv \lim_{n\to\infty}{\frac{1}{n}\ln{\|{\Phi_n(E)}\|}}
=
\frac{1}{\widetilde{q}}\ln[{\mathrm{Spr\,}(\Phi_{\widetilde{q}}(E))}],
\end{equation*}
where $\mathrm{Spr\,}(\cdot)$ is the spectral radius and
\begin{multline*}
\Phi_n(E) = T_n(E) \cdots T_2(E) T_1(E)~, \quad
T_j(E) = 
\left(\begin{array}{cc}
E - V(j) & -1 \\
  1      &   0  \\
\end{array}\right)~.\end{multline*}
 Since $\det{\Phi_n(E)} = 1$ (for all $n$) we have that
for
\begin{equation}\label{eq0}
E \in \mathcal{A}\equiv{\{E\in{\mathbb R}:\gamma(E) = 0\}},
\end{equation}
$\Phi_n(E)$ has eigenvalues presentable as
$e^{\pm{ik(E){\widetilde{q}}}},$\ with $k(E)\in\mathbb{R}$  ($k(E)$\ is the so
called Bloch wave number). For $E \notin \mathcal{A}$ the
eigenvalues of $\Phi_n(E)$ are given by $e^{\pm \gamma(E) {\widetilde{q}}}$, where
$\gamma(E)$ is the Lyapunov exponent.

The corresponding eigenvectors give rise (as initial conditions) to
the solutions $\{\psi^\pm(n)\}_{n \in \N}$ of the equation $H\psi =
E \psi$ that satisfy
\[ \psi^\pm(n+{\widetilde{q}}) =
    \begin{cases}
        e^{\pm \pi i k {\widetilde{q}}} \psi^\pm (n)~, &E \in \mathcal{A}~, \\
        e^{\pm \gamma {\widetilde{q}}} \psi^\pm (n)~, &E \notin \mathcal{A}~.
    \end{cases} \]
Therefore  the operator $H - E$ has
bounded inverse if and only if $E \notin \mathcal{A}$, hence $\sigma
= \mathcal{A}$.

We define the restriction $H^{[k,\infty)}$ of a periodic operator $H$ to
an interval $[k,\infty)\subset\mathbb Z$ as follows:
\begin{equation*}\begin{split}
(H^{[k,\infty)}\psi)(n) &= \psi(n + 1) + \psi(n - 1) + V(n)\psi(n) 
\quad \textrm{for $n
> k$},\\
(H^{[k,\infty)}\psi)(k) &= \psi(k + 1) + V(k)\psi(k)~.
\end{split}
\end{equation*}

Doing this allows us to use information on sub-intervals to gain
information about $H$\ on $\mathbb{Z}$.\ Define
the  Green functions of $H^{[k,\infty)}$ via
\begin{equation}\label{eqGakinfty}
G^{[k,\infty)}(n,m;z)\equiv{\langle{\delta_m,(H^{[k,\infty)} - z
)^{-1}\delta_n}\rangle}~,
\end{equation}
where $\delta_j$\ is the vector in $\ell^2(\mathbb{Z})$,\ with
entries $\delta_j(n)$\ that are $1$\ if $n = j$\ and $0$\ if
$n\neq{j}$. 
From the second resolvent identity we have:
\begin{equation}\label{eqResolventIdentity}
G^{[k,\infty)}(k,l;E) =
G^{[k,\infty)}(k,n;E)G^{[n+1,\infty)}(n+1,l;E)
\end{equation}
for $k < n < l$. For a periodic potential of period
$\widetilde{q}$ and for $G^{[1, \infty)}(1,n; E)$ being the
appropriate Green function, iterative application of
(\ref{eqResolventIdentity}) yields:
\begin{equation}\label{eq:**}\begin{split}
&G^{[1, \infty)}(1, m\widetilde{q}; E) \\
 &\quad= G^{[1, \infty)}(1, \widetilde{q}; E) \,
   G^{[\widetilde{q}+ 1, \infty)}(\widetilde{q}+1, 2\widetilde{q};
    E) \\
    &\qquad\qquad\cdots \,
   G^{[(m-1)\widetilde{q}+1, \infty)}((m-1)\widetilde{q}+1, m\widetilde{q};
    E)\\
 &\quad= \left(G^{[1, \infty)}(1, \widetilde{q}; E)\right)^m~,
\end{split}\end{equation}
since
\[ G^{[(k-1)\widetilde{q}+1, \infty)}((k-1)\widetilde{q}+1, k\widetilde{q}; E)
    = G^{[1, \infty)}(1, \widetilde{q}; E) \]
by periodicity.

\vspace{2mm} An additional elementary connection, which is very
useful in proving our result, is:
\begin{equation*}
\gamma(E) =  -\lim_{n\to\infty}\frac{1}{n}\ln{|G(1,n;E)|},
\end{equation*}
and (\ref{eq:**}) implies that in the periodic case (period
$\widetilde{q}$)
\begin{equation}\label{grenn_vs_lyapunov}
\gamma(E) =  -\frac{1}{\widetilde{q}}\ln|G^{[1, \infty)}(1,
\widetilde{q}; E)|.
\end{equation}
Therefore, for $E\notin{\mathcal{A}}$,\ we have
\[
|G^{[1, \infty)}(1, \widetilde{q}; E)| = e^{-\gamma{\widetilde{q}}}.
\]

In order to estimate the Lyapunov exponent for the system
$(\ref{eigenvalue_eq})$  (and we will explain later how to get rid of
the dependence on $\theta$ for the operator (\ref{almostMathieu})), we move $E$\ slightly off the real
axis, estimate the Lyapunov exponent for the resulting complex
$E$'s\ and then finally relate it to the Lyapunov exponent on the
real line. For this last step, we use the following general lemma due to
Surace \cite{surace}:
\begin{lemma}[Surace]\label{surace_lemma}
For any $\epsilon,\eta > 0$,
\begin{equation*}
\operatorname{meas} \{ E :  |\gamma(E + i\epsilon) - \gamma(E)|\geq{\eta}\}
\leq\frac{\pi\epsilon}{\eta}~.
\end{equation*}
\end{lemma}

The reason for making $E$\ complex is to eliminate the possibility
of extremely small divisors in the Green function
$G(\cdot,\cdot;E)$.\ Thus, we begin with modifying the system
$(\ref{eigenvalue_eq})$,\ and from now, we study the system
\begin{equation}\label{eqcomplexsys}
\psi(n + 1) + \psi(n - 1) + V(n)\psi(n) = (E +
i\epsilon)\psi(n).
\end{equation}

For future reference, we recall two rough bounds:  for  every interval $I\subset\mathbb{Z}$,
\begin{equation}\label{eqboundforgreenfunctiononremainingintervals}
\|{G^I}\| = \|{(H^I - E -
i\epsilon)^{-1}}\|\leq{\frac{1}{\epsilon}},
\end{equation}
and, if $I \ni k_0$,
\begin{equation}\label{eq4}
\sum_{k\in I}|G^{I}(k_0, k; E+i\epsilon)|^2 = \frac{1}{\epsilon} \mathrm{Im\,}G^{I}(k_0,k_0;E + i\epsilon) \leq{\frac{1}{\epsilon^2}}.
\end{equation}

In the spectral analysis of periodic operators (of period
$\widetilde{q}$) a key role is played by the discriminant $D(E)$,
defined by:
\begin{equation}\label{discriminant}
D(E)\equiv \operatorname{Tr}(\Phi_{\widetilde{q}}(E)).
\end{equation}
$D(E)$\ is a polynomial in $E$\ of order $\widetilde{q}$\ with the
following properties (see, e.g., \cite{last2}):
\begin{itemize}
\item[$\openrm \mathrm{1}\closerm$] $D(E)$\ has $\widetilde{q}$\ real simple
zeroes.
\item[$\openrm \mathrm{2}\closerm$] $D(E) \geq 2$\ at all its maxima points, and
$D(E) \leq -2$\ at all its minima points.
\end{itemize}
\noindent The spectrum is precisely the set $\{ E: -2 \leq D(E) \leq
2\}$.\ Therefore, from the properties of $D(E)$\ it is clear that
the spectrum is the union of $\widetilde{q}$\ closed intervals
(bands), such that $D(E)$\ is strongly monotone on each band.

\medskip
\paragraph{\bf The periodic almost Mathieu operator}
Up until now we discussed general periodic operators. Now we concentrate on the almost Mathieu case.  

Consider $H_{\alpha, 2, \theta}$, where
$\alpha = \frac{p}{q}$,\ $p, q\in\mathbb{N}$, and we assume that $p
< q$\ are relatively prime. In this case the spectrum
$\sigma(\alpha, 2, \theta)$\ does depend on $\theta$, and we will
be interested in the two spectral sets:
\begin{equation}\label{union_intersection}
\begin{array}{ll}
S(\alpha, \lambda) \equiv \bigcup_{\theta} {\sigma(\alpha, \lambda,
\theta)}, & \textrm{\ } S_{-}(\alpha, \lambda) \equiv
\bigcap_{\theta} {\sigma(\alpha, \lambda, \theta)}.
\end{array}
\end{equation}
These sets are also well defined in
the case that $\alpha\in\mathbb{R\setminus{Q}}$, in which case we have $S(\alpha, \lambda) =
S_{-}(\alpha, \lambda) = \sigma(\alpha, \lambda, \theta)$. 

In the almost Mathieu case  the $\theta$\
dependence of $D_{p/q, \lambda, \theta}(E)$ is described by the following 
formula, due to Chambers \cite{chembers}.
\begin{prop}[Chambers]\label{prop1}
If $p$,\ $q$\ are relatively prime, then:
\begin{equation}\label{chamber_formula}
D_{p/q, \lambda, \theta}(E) = \Delta_{p/q, \lambda}(E) -
2\left(\frac{\lambda}{2}\right)^q\cos{\theta{q}},
\end{equation}
where $\Delta_{p/q, \lambda}(E)\equiv D_{p/q, \lambda,
\frac{\pi}{2q}}(E)$.\
\end{prop}

\noindent From this formula 
\[
S\left(\frac{p}{q}, \lambda\right)= \{ E : |\Delta_{\frac{p}{q}, \lambda}(E)|
\leq 2 + 2(\lambda/2)^q \}.
\]
\noindent Moreover, one can see from (\ref{chamber_formula}) that if
$\lambda > 2$\ then $S_{-}(\frac{p}{q}, \lambda) = \varnothing$,\ and
if $\lambda \leq 2$\ then
\[
S_{-}\left(\frac{p}{q}, \lambda\right) = \{ E: |\Delta_{p/q, \lambda}(E)| \leq
2 - 2(\lambda/2)^q\}.
\]

\noindent From the Chambers formula combined with the properties of $D(E)$
 stated above,  it follows that $\Delta_{\frac{p}{q},
\lambda}(E) \geq 2 + 2(\lambda/2)^q$\ at all its maxima points, and
$\Delta_{\frac{p}{q}, \lambda}(E) \leq -2 - 2(\lambda/2)^q$\ at all
its minima points. In addition, each of the sets $S(\frac{p}{q},
\lambda)$\ and $S_{-}(\frac{p}{q}, \lambda)$,\ when it's not empty,
is the union of $q$\ closed intervals (bands), such that
$\Delta_{p/q, \lambda}(E)$\ is strongly monotone on each band.
Also note that for $\lambda = 2$  (\ref{chamber_formula}) is consistent with (\ref{eqdefdist}).

An important feature of the set $\specU$ is its H\"older continuity
in $\alpha$, which allows to study the set $\specU$\ for an
irrational $\alpha$\  via $S(\alpha_i,
\lambda)$, for rational $\alpha_i \to \alpha$.
\begin{prop}[Avron--van Mouche--Simon \cite{simon_vanmouche_avron}]\label{continuity_of_spectra}
For every $\lambda > 0$, if
$|\alpha - \alpha^\prime|$ is sufficiently small, then for every $E\in\specU$,\
there is $E^\prime\in{S(\alpha^\prime,\lambda)}$\ with:
\[
|E - E^\prime| < 6(\lambda|\alpha - \alpha^\prime|)^{\frac{1}{2}}.
\]
\end{prop}

We will also use the following result of Last and Wilkinson
\cite{last_wilkinson} regarding the special structure of the almost
Mathieu operator for $\lambda = 2$:
\begin{lemma}[Last--Wilkinson]\label{last_lemma}
If $p$,\ $q$\ are relatively prime, then:
\begin{equation}\label{last_wilk_estim_for_derivative}
\sum_{n = 1}^q{\frac{1}{|\Delta_{\frac{p}{q}, 2}^\prime(E_n)|}} =
\frac{1}{q},
\end{equation}
where $\Delta_{\frac{p}{q}, \lambda}^\prime(E) \equiv
\frac{d}{dE}\Delta_{\frac{p}{q}, \lambda}(E)$,\ and $E_1, E_2,
\ldots, E_q$\ are the $q$\ zeroes of $\Delta_{\frac{p}{q}, 2}(E)$.\
\end{lemma}

\medskip
\paragraph{\bf Hausdorff dimension}
Recall \cite{Falconer} that the Hausdorff dimension of a set
$S\subset\mathbb{R}$\ is
given by:
\begin{equation}\label{hausdorff_dimension}
\dimension(S) = \inf \left\{ t\in\mathbb{R^{+}} \, \big{|} \, 
\lim_{\delta\rightarrow 0}\inf_{\delta
\textrm{-covers}}\sum_{n}(\meas(U_n))^t < \infty \right\},
\end{equation}
where a $\delta$-cover $S\subset\bigcup_{n=1}^\infty{U_n}$ is
a cover  such that every $U_n$\ is an
interval of length smaller than $\delta$.\  The next lemma 
is used in the proof of Theorem
$\ref{mainTheorem}$.\

\begin{lemma}\label{hausdorff_dim_lemma}
Let $S\subset\mathbb{R}$, and let ${q_n}, {\widetilde{q}_n}$ be two sequences of natural numbers, such
that for every $n \in \N$ one has a cover of $S$ by intervals:
\[
S \subset  \left(\bigcup_{i=1}^{q_n}I_{i,n}\right)\bigcup
\left(\bigcup_{j=1}^{\widetilde{q}_n}{\widetilde{I}_{j,n}}\right),
\]
so that
\begin{enumerate}
\item $q_n \to \infty$, $\widetilde{q}_n \to \infty$;
\item $\meas(\bigcup_{i=1}^{q_{n}}I_{i,n}) < \frac{C_1}{q_{n}^{\beta_1}},
\textrm{\ } \meas(\bigcup_{j=1}^{\widetilde{q}_n}\widetilde{I}_{j,n}) <
\frac{C_2}{\widetilde{q}_{n}^{\beta_2}}$,
\end{enumerate}
where $C_1,C_2, \beta_1, \beta_2$ are positive constants. Then
\[
\dimension(S) \leq \max\left(\frac{1}{1 + \beta_1}, \frac{1}{1 +
\beta_2}\right).
\]
\end{lemma}

\begin{proof} Let $t = \max({\frac{1}{1 + \beta_1}, \frac{1}{1 + \beta_2}})$,\ then
by Jensen's inequality:
\begin{equation*}
\frac{1}{q_{n}}\sum_{i = 1}^{q_{n}}{(\meas(I_{i,n}))^t} \leq
\left(\frac{1}{q_{n}}\sum_{i = 1}^{q_{n}}{\meas(I_{i,n}})\right)^t,
\end{equation*}
\begin{equation*}
\frac{1}{\widetilde{q}_{n}}\sum_{j = 1}^{\widetilde{q}_{n}}{(\meas(\widetilde{I}_{j,n}))^t}\leq
\left(\frac{1}{\widetilde{q}_{n}}\sum_{j = 1}^{\widetilde{q}_{n}}{\meas(\widetilde{I}_{j,n}})\right)^t,
\end{equation*}
which implies:
\begin{equation}\label{estim_covers_haus_dim}
\begin{split}
&\sum_{i = 1}^{q_{n}}{(\meas(I_{i,n}))^t}
    + \sum_{j =1}^{\widetilde{q}_{n}}{(\meas(\widetilde{I}_{j,n}))^t}  \\
&\qquad\leq (\max(C_1, C_2))^t(q_{n}^{1 - t(1 + \beta_1)} + \widetilde{q}_{n}^{1 - t(1 + \beta_2)}) \leq C_t~, 
\end{split}
\end{equation}
where $C_t = (\max(C_1, C_2))^t$.
Since $q_{n}, \widetilde{q}_n\rightarrow\infty$ as $n\rightarrow\infty$ we
obtain from $(\ref{estim_covers_haus_dim})$:\
\[\dimension(S) \leq
\max\left(\frac{1}{1 + \beta_1}, \frac{1}{1 + \beta_2}\right).~\]
\end{proof}

\section{Proof of Theorem $\ref{matrixTheorem}$}\label{sec:matrixTheorem}
We focus on the first inequality, the second one is proved by a similar argument.
Define 
\[ \psi(1) = \phi_1^+ \quad \text{and} \quad {\psi}(n) = \Phi_{n-1} {\psi} (1)~, 
\quad 2 \leq n \leq N+1~,\]
and represent
${\psi} (n)$ as a linear combination:
\[\psi(n) = A_n \phi^+_n + B_n \phi^-_n~, \] 
where we set $T_{N+1} = T_N$.
Then $\binom{A_1}{B_1} = \binom{1}{0}$.  From the definition of $T_n$, we have
\begin{equation}\label{nextstepwithT_i}
\psi(n+1) = T_n\openrm
{A_n\phi^{+}_n +
B_n\phi^{-}_n}\closerm =
A_n{e^{\gamma_{n} + i\zeta_{n}}}\phi^{+}_n +
B_n{e^{-\gamma_{n} - i\zeta_{n}}}\phi^{-}_n~.
\end{equation}
Comparing (\ref{nextstepwithT_i}) with
\begin{equation}\label{nextstep}
\psi(n+1) =
A_{n+1}\phi^{+}_{n+1} +
B_{n+1}\phi^{-}_{n+1},
\end{equation}
and denoting
\[ U_n \equiv
\left(\begin{array}{cc}
\phi_n^+(1) & \phi_n^-(1) \\
\phi_n^+(2) & \phi_n^-(2) \\
\end{array}\right)~, \quad \Lambda_n \equiv \left(\begin{array}{cc}
e^{\gamma_n + i\zeta_n} & 0 \\
  0      &   e^{-\gamma_n - i\zeta_n}  \\
\end{array}\right)~, \]
we get
\[ U_{n+1} \binom{A_{n+1}}{B_{n+1}} =  U_n \Lambda_n \binom{A_n}{B_n}~,\]
and
\[ \binom{A_{n+1}}{B_{n+1}} = \left[ U_{n+1}^{-1} (U_n - U_{n+1}) + I \right] \Lambda_n \binom{A_n}{B_n}~.\]
The matrix $\rho_n = U_{n+1}^{-1} (U_n - U_{n+1})$ satisfies:
\[\begin{split}\|\rho_n\| &\leq \frac{\|U_{n+1}\|}{|\det U_{n+1}|} \| U_n - U_{n+1} \| \\
&\leq \frac{2 \max{\openrm \| \phi_{n}^{+} - \phi_{n+1}^{+}\|, \|
\phi_{n}^{-} - \phi_{n+1}^{-}\|\closerm}}{|\det U_{n+1}|} \leq \frac{\beta}{2}\, \gamma_n{e^{-\gamma_n}}
\end{split}\]
where the second inequality follows from the fact that the norm of every
row of $U_{n+1}$ is equal to one, and the third one,-- from assumption (3) of the
theorem.

We prove by induction that for every $n$, $|B_n| \leq \beta |A_n|$. From
the definition $|B_1| = 0 < \beta = \beta |A_1|$. Assume that $|B_n| \leq
\beta |A_n|$ for some $n > 0$. From the relations above we obtain
\begin{align}\label{induction}\begin{split}
|A_{n+1}| &\geq e^{\gamma_n}|A_n| -
e^{\gamma_n}\gamma_n{e^{-\gamma_n}}{\frac{\beta}{2}}|A_n| -
e^{-\gamma_n}\gamma_n{e^{-\gamma_n}}{\frac{\beta}{2}}|B_n| \\
&\geq |A_n| \left( e^{\gamma_n} - \frac{\beta}{2} \gamma_n - \frac{\beta^2}{2} \gamma_n e^{-2 \gamma_n} \right)~,\end{split}\\\label{indbis}\begin{split}
|B_{n+1}|&\leq e^{-\gamma_n}|B_n| +
e^{-\gamma_n}\gamma_n{e^{-\gamma_n}}{\frac{\beta}{2}}|B_n| +
e^{\gamma_n}\gamma_n{e^{-\gamma_n}}{\frac{\beta}{2}}|A_n| \\
&\leq |A_n| \left( \beta e^{-\gamma_n} + \frac{\beta^2}{2} \gamma_n
e^{-2 \gamma_n} + \frac{\beta}{2} \gamma_n \right)~.
\end{split}
\end{align}
An elementary computation using that $e^x \geq 1+x$ for any $x \in \mathbb{R}$ yields
\[ \beta e^{-\gamma_n} + \frac{\beta^2 \gamma_n}{2} e^{-2 \gamma_n} +
\frac{\beta \gamma_n}{2} \leq \beta \left\{ e^{\gamma_n} - \frac{\beta \gamma_n}{2} -
\frac{\beta^2 \gamma_n}{2} e^{-2 \gamma_n}\right\}~, \] 
hence $|B_{n + 1}| \leq \beta  |A_{n + 1}|$, concluding the
induction.

\vspace{2mm}\noindent By iterative application of $\openrm
\ref{induction}\closerm$ and the elementary inequality 
\[ \ln(1 - x e^{-x}) \geq -x~, \quad x \geq 0~, \]
we obtain
\begin{equation}\label{eq:4}
\begin{split}|A_{N+1}|&\geq e^{\sum_{n = 1}^N{\gamma_n}}\prod_{n = 1}^N\left(1 -
\frac{\beta}{2}\gamma_n(\beta e^{-3\gamma_n} + e^{-\gamma_n})\right) \\&\geq
e^{\sum_{n=1}^N \gamma_n} \prod_{n=1}^N \left[ 1 - \beta \gamma_n
e^{- \beta \gamma_n} \right]\geq e^{(1 - \beta){\sum_{n=1}^N {\gamma_n}}}~.
\end{split}\end{equation}
The proof of the theorem is concluded by
the estimate
\[ \| \Phi_N \phi_1^+ \| 
    \geq |A_{N+1}| \|\phi_{N+1}^+\| - |B_{N+1}| \|\phi_{N+1}^-\| \geq |A_{N+1}|(1 - \beta)~.\]
\qed

\section{Proof of Theorem $\ref{lebesgueMeasure}$}\label{sec:lebesgueMeasure}
\subsection{Outline of the proof}
Fix $q \in \N$ and consider the almost Mathieu operator
$H_{p/q,2,\theta}$, where $p < q$ are relatively prime. Let $\delta
> 0$, and assume that $\widetilde{p} < \widetilde{q}$ are also
relatively prime, and that $\left|
\frac{\widetilde{p}}{\widetilde{q}} - \frac{p}{q} \right|$ is
sufficiently small. Consider the almost Mathieu operator
$H_{\widetilde{p}/\widetilde{q}, 2, \theta}$. Let $\Delta_{p/q, 2}$
be as defined in Chambers' formula (\ref{chamber_formula}), and
$S(\frac{\widetilde{p}}{\widetilde{q}}, 2)$ as defined in
Section~\ref{s:am.prel}. Consider
\[ J_\delta = J_\delta^{(1)} = \left\{ E \, | \, \left| \Delta_{p/q, 2}(E) \right| > \delta > 0 \right\}~. \]

The Lyapunov exponent for the operator
$H_{\widetilde{p}/\widetilde{q}, 2, \theta}$ depends on $\theta$.
However, Chambers' formula allows us to relate the set
$S(\frac{\widetilde{p}}{\widetilde{q}}, 2)$ to the Lyapunov exponent
corresponding to $\theta = 0$ (which we denote $\gamma(E)$), as follows:

\begin{lemma}
For every $E \in S(\frac{\widetilde{p}}{\widetilde{q}}, 2)$,
$\gamma(E) \leq \frac{\ln 6}{\widetilde{q}}$.
\end{lemma}

\begin{proof}
For every $E \in S(\frac{\widetilde{p}}{\widetilde{q}}, 2)$, Proposition~\ref{prop1} yields $\left| \Delta_{\widetilde{p}/\widetilde{q},
2}(E) \right| \leq 4$, 
and hence $\left|
D_{\widetilde{p}/\widetilde{q}, 2, 0}(E) \right| \leq 6$. Thus for
 $E$ such that $\gamma(E) \neq 0$ we have
\[ 6 \geq \left| D_{\widetilde{p}/\widetilde{q}, 2, 0}(E) \right|
    = e^{\gamma(E)\widetilde{q}} + e^{-\gamma(E)\widetilde{q}} \geq e^{\gamma(E)\widetilde{q}}~.\]
\end{proof}

Therefore, from now on we work with $\theta = 0$.
In order to estimate the Lyapunov exponent for the system
$(\ref{eigenvalue_eq})$,\ where $V(n) =
2\cos(2\pi\frac{\widetilde{p}}{\widetilde{q}}n)$, we move $E$\
slightly off the real axis, estimate the Lyapunov exponent for the
resulting complex $E$'s and then finally relate it to the Lyapunov
exponent on the real line. The reason for making $E$\ complex is to
eliminate the possibility of extremely small divisors in the Green
function $G(\cdot, \cdot;E)$,\ corresponding to the operator
$H_{\widetilde{p}/\widetilde{q}, 2, 0}$. Thus, we begin with
modifying the system $(\ref{eigenvalue_eq})$, and from now on, we
will study the system
\begin{equation}\label{eqcomplexsys'}
\psi(n + 1) + \psi(n - 1) + V(n)\psi(n) = (E + i\epsilon)\psi(n).
\end{equation}
The main purpose in our proof is to show that for $\epsilon \cong
e^{-\frac{c\delta\widetilde{q}}{q}}$,\
where $c$ is a sufficiently small constant, the Lyapunov exponent $\gamma(E + i\epsilon)$\ of the
modified system is at least of order $c \delta/q$.\ As
follows from Lemma~\ref{surace_lemma}, the Lyapunov exponent
$\gamma(E)$\ of the system $(\ref{eigenvalue_eq})$\ is close to
the Lyapunov exponent $\gamma(E + i\epsilon)$\ of the system
$(\ref{eqcomplexsys})$,\ except for values of $E$\ which lie in a
set of Lebesgue measure $\cong \epsilon$.\ Therefore, we
just need to show that, for $\epsilon \cong
e^{-\frac{c \delta\widetilde{q}}{q}}$,
the Lyapunov exponent of the system $(\ref{eqcomplexsys})$\ is
larger than or equal to $c \delta/(2q)$.\

Following is a brief explanation of why the Lyapunov exponent of the
modified system is large, skipping technical details. The technical
proof is given later in this Section.

Fix $E\in J_{\delta}$,\ for which $\Delta_{p/q,2}(E) < 0$,\ where
$\Delta_{p/q,2}(E)$\ is defined in $(\ref{chamber_formula})$.\ We
prove that the Green function $G^{[1, \infty)}(\cdot, \cdot, E +
i\epsilon)$\ decays exponentially fast for all such $E$'s.\ The
proof for the case $\Delta_{p/q,2}(E) > 0$\ is very similar (one
should work with $\theta = \frac{\pi}{q}$ rather than $\theta = 0$),
hence we prove that the Green function $G^{[1, \infty)}(\cdot,
\cdot, E + i\epsilon)$,\ that corresponds to $H^{[1,
\infty)}_{\widetilde{p}/\widetilde{q}, 2, \theta}$, decays
exponentially fast for all $E\in J_{\delta}$.\

In order to prove that, we introduce an {\em intermediate problem}, which
 interpolates between the periodic
potential corresponding to $p/q$ and that corresponding to $\widetilde{p}/\widetilde{q}$:
\[ (\widetilde{H}\psi)(n) = \psi(n + 1) + \psi(n - 1) + \widetilde{V}(n)\psi(n) ~, \]
where 
\begin{equation}\label{periodic_approximation}
\widetilde{V}(n) = \left\{ \begin{array}{ll}
2\cos(2\pi\frac{p}{q}n + \theta_n) & \textrm{\ if\ } n < l_0q\\
2\cos(2\pi\frac{p}{q}n + \theta_{l_0q}) & \textrm{\ if\ } n \geq
l_0q,
\end{array}\right.
\end{equation}
$0 < l_0 \leq \widetilde{q}$\ is an integer chosen by $l_0 =
[\frac{\widetilde{c}\widetilde{q}\sqrt{\delta}}{q}]$,\ where $\widetilde{c}$\ is
some universal constant, and
\[ \theta_n = \theta + 2\pi\left(\frac{\widetilde{p}}{\widetilde{q}} - \frac{p}{q}\right)n~. \]

Our choice of $l_0$\ will ensure that for all
$\widetilde{\theta}\in[0, \frac{2\pi{l_0}}{\widetilde{q}}]$, $E$ is sufficiently
far from the spectrum of $H_{\frac{p}{q}, 2,
\widetilde{\theta}}$.

On a large initial interval $\widetilde{V}$ coincides with
$V_{\frac{\widetilde{p}}{\widetilde{q}}, 2, 0}$, and afterwards coincides with
$V_{\frac{p}{q}, 2, \theta_{l_0q}}$. If the initial interval is large enough, the intermediate
operator is close to $H_{\frac{\widetilde{p}}{\widetilde{q}}, 2, 0}$.
On the other hand, if we regard the potential on the initial interval as $V_{\frac{p}{q}, 2, \theta}$
with slowly varying $\theta$, we can connect the spectrum of the intermediate problem
to the intersection of the spectra of $H_{\frac{p}{q}, 2, \theta}$, which is exactly
$S_-(\frac{p}{q}, 2)$.

In the language of Green functions, this works as follows. First, we show that the Green function
corresponding to the intermediate problem is close to that corresponding to
$H_{\frac{\widetilde{p}}{\widetilde{q}}, 2, 0}$ (since the two potentials coincide on a long interval).
The intermediate problem is eventually periodic of period $q$ (with constant $\theta$), and this allows us to
estimate its Green function for energies which are sufficiently far from $S_-(\frac{p}{q}, 2)$.
From these two steps we obtain an estimate on the Green function corresponding to
$H_{\frac{\widetilde{p}}{\widetilde{q}}, 2, 0}$ (and eventually $ H_{\frac{\widetilde{p}}{\widetilde{q}}, 2, \theta} $).

\subsection{Proof of Theorem~\ref{lebesgueMeasure}}

We start with the case $J_\delta = J_\delta^{(1)}$; also, we
assume that $0 < \delta \leq 1$ (since the spectrum is bounded, the general case follows by adjusting the numerical
constants).
Fix $E\in J_{\delta}$. We
prove that the Green function $G^{[1, \infty)}(\cdot, \cdot, E +
i\epsilon)$\ corresponding to the almost Mathieu operator $H^{[1,
\infty)}_{\widetilde{p}/\widetilde{q}, 2, 0}$, decays exponentially
fast for all such $E$'s.\ We focus on the case
 $\Delta_{p/q,2}(E) < 0$; the case $\Delta_{p/q,2}(E)
> 0$\ is very similar.

Consider the  operator $\widetilde{H} = \Delta + \widetilde{V}$, where,
as in (\ref{periodic_approximation}),
\begin{equation}\label{middleProblem}\widetilde{V}(n) = \left\{
\begin{array}{ll}
2\cos(2\pi{\frac{p}{q}}n + \theta_n) & \textrm{ for $n < l_0q$}\\
2\cos(2\pi{\frac{p}{q}}n + \theta_{l_0q}) & \textrm{ for $n \geq l_0q$},
\end{array}\right.
\end{equation}
 $0 < l_0 \leq \widetilde{q}$\ is an integer that will be
determined later, and
\[
\theta_n = 2\pi\left(\frac{\widetilde{p}}{\widetilde{q}} -
\frac{p}{q}\right)n~.
\]
We will verify
that our choice of $l_0$ guarantees that for all $\widetilde{\theta}\in[0,
\frac{2\pi{l_0}}{\widetilde{q}}]$, the energy $E$ is sufficiently far
from the spectrum of $H_{\frac{p}{q}, 2, \widetilde{\theta}}$.

Let $\widetilde{G}^{[\cdot, \infty)}(\cdot, \cdot; E + i\epsilon)$\
be the Green function corresponding to the operator $\widetilde{H}$.\
For convenience, we denote from
now on $G^{[\cdot, \cdot)}(\cdot, \cdot; E + i\epsilon) \equiv
G^{[\cdot, \cdot)}(\cdot, \cdot)$.\ We want to obtain an estimate on
$|G^{[1,\infty)}(1, q\widetilde{q})|$.\ Assume for a moment that the
following holds:
\begin{equation}\label{final_estim}
|G^{[1,\infty)}(1, q\widetilde{q})| \leq
\frac{20}{\epsilon^4} e^{-\frac{c'\delta\widetilde{q}}{q}}~,
\end{equation}
where $c'>0$ is a numerical constant. For $\epsilon = 20^{1/4} e^{-\frac{c'\delta\widetilde{q}}{8 q}}$
we obtain 
\[ |G^{[1,\infty)}(1, q\widetilde{q})| \leq
e^{-\frac{c'\delta\widetilde{q}}{2q}}~,
\]
whence, using $(\ref{grenn_vs_lyapunov})$, 
$\gamma(E + i\epsilon)\geq \frac{c'\delta}{2q}$.    From
Lemma $\ref{surace_lemma}$\ we have
\begin{equation*}
\begin{split}
\meas\,\left(\left\{E : |\gamma(E + i\epsilon) - \gamma(E)|\geq
\frac{c'\delta}{4q}\right\}\right) \leq
\frac{4\pi q \epsilon}{c'\delta} \leq \frac{C q}{\delta} e^{-\frac{c'\delta\widetilde{q}}{8 q}}
\end{split}
\end{equation*}
therefore, if $\widetilde{q} \geq 4q/(C\delta)$,
\begin{equation*}
\begin{split}
&\meas\,\left(\left\{E : \gamma(E)\leq
\frac{\ln6}{\widetilde{q}}\right\}\cap{J_{\delta}}\right) \leq{\meas\,\left(\left\{E :
\gamma(E)\leq{\frac{c'\delta}{4q}}\right\}\cap{J_{\delta}}\right)} \\
&\leq{\meas\,\left(\left\{E : |\gamma(E + i\epsilon) -
\gamma(E)|\geq{\frac{c'\delta}{4q}}\right\}\cap{J_{\delta}}\right)}\leq
\frac{C q}{\delta} e^{-\frac{c'\delta\widetilde{q}}{8 q}}~,\end{split}
\end{equation*}
and the proof of the theorem is concluded. 

\noindent We prove $(\ref{final_estim})$\ in two steps:
\begin{enumerate}
\item[(i)] for any $l_0$, $|G^{[1,\infty)}(1, q\widetilde{q})| \leq
\frac{5}{{\epsilon}^3}|\widetilde{G}^{[1, \infty)}(1, l_0{q})|$,
\item[(ii)]for  $l_0 =
[\frac{c'\sqrt{\delta}\widetilde{q}}{q}]$, one has $|\widetilde{G}^{[1, \infty)}(1, l_0{q})|\leq
\frac{4}{\epsilon}e^{-\frac{c'\delta\widetilde{q}}{q}}$.
\end{enumerate}

\vspace{4mm}\noindent Proof of (i): From
(\ref{eqResolventIdentity}) and the definition of
$\widetilde{V}$\ we obtain
\begin{equation*}
\begin{split}
&\Big|G^{[1,\infty)}(1, q\widetilde{q})\Big| \\
&\quad= \Big|G^{[1,\infty)}(1,
l_{0}q)G^{[l_{0}q + 1, \infty)}(l_{0}q + 1, q\widetilde{q})\Big| \\
&\quad\leq \Big|\widetilde{G}^{[1, \infty)}(1, l_{0}q)G^{[l_{0}q +
1,
\infty)}(l_0q + 1, q\widetilde{q})\Big| \\
&\qquad+ \Big|\left(G^{[1,\infty)}(1, l_{0}q) - \widetilde{G}^{[1,
\infty)}(1, l_{0}q)\right)\Big|\Big|G^{[l_{0}q + 1, \infty)}(l_{0}q
+ 1,
q\widetilde{q})\Big| \\
&\quad= \Big|G^{[l_{0}q + 1,
\infty)}(l_0q + 1, q\widetilde{q})\Big| \\
&\qquad\times \left[ \Big| \widetilde{G}^{[1, \infty)}(1, l_{0}q)
\Big|
    + \Big|\left(G^{[1,\infty)}(1, l_{0}q) - \widetilde{G}^{[1,
\infty)}(1, l_{0}q)\right)\Big| \right]~.
\end{split}
\end{equation*}
First, we estimate the second term of the sum. From the second resolvent
identity we obtain
\begin{equation}\label{difference_estim}
\begin{split}
&\Big|G^{[1,\infty)}(1, l_{0}q) - \widetilde{G}^{[1, \infty)}(1, l_{0}q)\Big| \\
&\quad= \Big|\sum_{k=1}^\infty{\widetilde{G}^{[1, \infty)}(1,
k)(\widetilde{V}(k) - V(k))G^{[1,\infty)}(k, l_{0}q)}\Big|\\
&\quad\leq \sum_{k=l_{0}q + 1}^\infty{\Big|\widetilde{G}^{[1,
\infty)}(1,
k)(\widetilde{V}(k) - V(k))G^{[1,\infty)}(k, l_{0}q)\Big|}\\
&\quad\leq 2\|V\|_{\infty}\sum_{k=l_{0}q +
1}^\infty{\Big|\widetilde{G}^{[1, \infty)}(1, k)G^{[1,\infty)}(k,
l_{0}q)\Big|}\\
&\quad\leq 4\sum_{k=l_{0}q + 1}^\infty{\Big|\widetilde{G}^{[1,
\infty)}(1, k)G^{[1,\infty)}(k, l_{0}q)\Big|}\equiv \widetilde{W}.
\end{split}
\end{equation}
\noindent Applying $(\ref{eqResolventIdentity})$  we see that for
$k\geq l_{0}q + 1$
\[
\widetilde{G}^{[1, \infty)}(1, k) = \widetilde{G}^{[1, \infty)}(1,
l_{0}q)\widetilde{G}^{[l_{0}q + 1, \infty)}(l_{0}q + 1, k).
\]
\noindent As a result we obtain
\[
\widetilde{W} \leq 4|\widetilde{G}^{[1, \infty)}(1,
l_{0}q)|\widetilde{A} \leq 4|\widetilde{G}^{[1, \infty)}(1,
l_{0}q)|\widetilde{B},
\]
where
\begin{align}\label{estim_sum}
\widetilde{A} &=& \sum_{k=l_{0}q + 1}^\infty{\Big|\widetilde{G}^{[l_{0}q +
1, \infty)}(l_{0}q + 1, k)G^{[1,\infty)}(k, l_{0}q)\Big|},\\
\label{estim_cauchy_schwartz}
\widetilde{B} &=& \left(\sum_{k=1}^{\infty}\Big|\widetilde{G}^{[1,
\infty)}(1, k)\Big|^2
\right)^{\frac{1}{2}}\left(\sum_{k=1}^{\infty}\Big|G^{[1, \infty)}(k,
l_{0}q)\Big|^2 \right)^{\frac{1}{2}},
\end{align}
and we used the
Cauchy-Schwartz inequality and the definition of $\widetilde{V}$.
From $(\ref{eq4})$\ we have
\begin{equation}\label{estim_sum_from_cauchy_shwartz_for_G}
\sum_{k=1}^{\infty}\Big|\widetilde{G}^{[1, \infty)}(1, k)\Big|^2 \leq
\frac{1}{\epsilon^2}~, \quad
\sum_{k=1}^{\infty}\Big|G^{[1, \infty)}(k, l_0q)\Big|^2 \leq
\frac{1}{\epsilon^2}~,\end{equation}
therefore, we obtain
\begin{equation}\label{final_difference_estim}
\Big|G^{[1,\infty)}(1, l_{0}q) - \widetilde{G}^{[1, \infty)}(1, l_{0}q)\Big|
\leq \Big|\widetilde{G}^{[1, \infty)}(1, l_{0}q)\Big|\frac{4}{\epsilon^2}.
\end{equation}
By combining $(\ref{final_difference_estim})$\ and $(\ref{eq4})$\ we
obtain
\begin{equation}\label{fde2}\begin{split}
&\Big|G^{[1,\infty)}(1, q\widetilde{q})\Big| \\
&\leq \Big|\widetilde{G}^{[1,
\infty)}(1, l_{0}q)\Big|\left(1 + \frac{4}{\epsilon^2}\right)\Big|G^{[l_{0}q
+ 1, \infty)}(l_0q + 1, q\widetilde{q})\Big| \\
&\leq \Big|\widetilde{G}^{[1, \infty)}(1,
l_{0}q)\Big|\frac{5}{\epsilon^3},
\end{split}\end{equation}
where the last inequality follows from the assumption that $\epsilon
< 1$.

\vspace{3mm}\noindent Proof of (ii). Now we estimate the
$|\widetilde{G}^{[1, \infty)}(1, l_{0}q)|$.\ We first verify that
for $l_0 =
[\frac{\widetilde{c}\sqrt{\delta}\widetilde{q}}{q}]$, where
${\widetilde{c}}> 0$, we have 
\begin{equation}\label{eq:iii}
\text{for all $0 \leq j \leq l_0 q$ and $E\in J_{\delta}$, \,\,\, 
$D_{p/q, 2, \theta_j}(E) < -2 - \frac{3\delta}{4}$.} 
\end{equation} Since
$E\in J_{\delta}$, $\Delta_{p/q, 2}(E) < -\delta$. From Chambers'
formula (\ref{chamber_formula})
\[
D_{p/q, 2, \theta_j}(E) = \Delta_{p/q, 2}(E) -
2\cos\theta_jq.
\]
Therefore it suffices to verify that
\[
2\cos\theta_{j}q \geq 2 - \frac{\delta}{4}, \quad 0 \leq j \leq l_0q~,
\]
and this holds when $\widetilde{c}$ in the definition of $l_0$ is sufficiently small.

\medskip
For $1 \leq k \leq q$, set
\[Q_{jq + k}^{-1}(E + i\epsilon)= \left(\begin{array}{cc}
0     &      1 \\
-1    &   E + i\epsilon - 2\cos(2\pi\frac{p}{q}(jq + k) + \theta_{jq + k})  \\
\end{array}\right)~,
\]
and
\begin{equation*}
\begin{split}
T_{j}^{-1}(E + i\epsilon) &= (Q_{jq + 1}^{-1}\ldots{Q_{(j
+1)q}^{-1}})(E + i\epsilon), \textrm{\ } 0 \leq j\leq l_0 - 1,
\\
\Phi^{-1}_{l_0q}(E + i\epsilon) &=
(T_{0}^{-1}\ldots{T_{l_0 - 1}^{-1}})(E + i\epsilon).
\end{split}
\end{equation*}
Then
\begin{equation}\label{matrix_eq_for_tildeG}
\left(\begin{array}{cc}
\widetilde{G}^{[1, \infty)}(1, 1)\\
1\\
\end{array}\right) = \Phi^{-1}_{l_0q}(E + i\epsilon)\left(\begin{array}{cc}
\widetilde{G}^{[1, \infty)}(1, l_{0}q + 1)\\
\widetilde{G}^{[1, \infty)}(1, l_{0}q)\\
\end{array}\right).
\end{equation}

Let us verify the assumptions of Theorem~\ref{matrixTheorem}.
According to the definition of $T^{-1}_j$, we have $\det(T^{-1}_j(E + i\epsilon)) = 1$.
Next, we claim that 
\begin{equation}\label{eq:qz}
|\mathrm{Tr} (T^{-1}_j(E + i\epsilon))| >  2 + \frac{\delta}{2}~.
\end{equation}
Indeed, define
\[\widehat{Q}_{jq + k}^{-1}(E + i\epsilon)= \left(\begin{array}{cc}
0     &      1 \\
-1    &   E + i\epsilon - 2\cos(2\pi\frac{p}{q}(jq + k) + \theta_{jq + 1})  \\
\end{array}\right)~,
\]
and
\begin{equation*}
\widehat{T}_{j}^{-1}(E + i\epsilon) = (\widehat{Q}_{jq + 1}^{-1}\ldots{\widehat{Q}_{(j
+1)q}^{-1}})(E + i\epsilon)~, \,  0 \leq j\leq l_0 - 1~.
\end{equation*}
The matrices $\widehat{T}_j$ are one-period transfer
matrices of the periodic almost Mathieu operator, therefore from (\ref{eq:iii})  $|\mathrm{Tr}(\widehat{T}^{-1}_j(E))| >  2 + \frac{3\delta}{4}$. To pass from $\widehat{T}_j^{-1}(E)$ to $T_j^{-1}(E+i\epsilon)$, we observe that
\[\begin{split}
\| \widehat{Q}_{jq + k}^{-1}(E + i\epsilon) - {Q}_{jq + k}^{-1}(E + i\epsilon) \|
&\leq 2 |\theta_{jq+k}-\theta_{jq+1}| \leq 4\pi\eta q~,\\
\|\widehat{Q}_{jq + k}^{-1}(E + i\epsilon)  - \widehat{Q}_{jq + k}^{-1}(E) \|
&= \epsilon = 20^{1/4} e^{- \frac{c''\delta\widetilde{q}}{q}} \leq 3 e^{- \frac{c''\delta}{q^2\eta}}~,
\end{split} \]
hence 
\[ \| \widehat{T}_j^{-1} (E) -  T_j^{-1} (E + i\epsilon) \| \leq (4\pi\eta q^2 + 3 e^{- \frac{c''\delta}{q^2\eta}})C_1^{q-1}\leq \eta C_2^q \leq \frac{\delta^2}{8} \leq \frac\delta8 \]
for sufficiently large $C$ in the definition of $\eta$, where we used the inequalities
\[ \| \widehat{Q}_{jq + k}^{-1}(E)\|~, \| {Q}_{jq + k}^{-1}(E + i\epsilon)\| \leq C_1~,
\quad \left| \frac{\tilde{p}}{\tilde{q}} - \frac{p}{q} \right|
\leq \eta~.\] 
Therefore
\[ |\mathrm{Tr} (T^{-1}_j(E + i\epsilon))| \geq2 + \frac{3\delta}{4}  - 2 \, \frac{\delta}{8} \geq 2  + \frac{\delta}{2}~,\]
as claimed in (\ref{eq:qz}). In particular, denoting the eigenvalues of the matrix $T^{-1}_j(E + i\epsilon)$
by $e^{\pm (\gamma_j + i\zeta_j)}$, we have:
\[ 
\gamma_j > \operatorname{arccosh} \left(1+ \frac{\delta}{4}\right) \geq c \sqrt{\delta}~.\]
Next, $\|T_{j}^{-1}(E + i\epsilon)\| \leq C^q$ and, assuming $\left| \frac{\tilde{p}}{\tilde{q}} - \frac{p}{q} \right|
\leq \eta$,
\begin{equation}\label{eq:normT}
\| T_{j}^{-1}(E + i\epsilon) - T_{j+1}^{-1}(E + i\epsilon)\| \leq 
C^q  q^2 \eta  \leq (2C)^q \eta~.
\end{equation}
Denoting by $\phi_j^\pm$ the unit eigenvectors of $T_{j}^{-1}(E+i\epsilon)$ and by $U_j$ the
matrix with columns $\phi_j^\pm$, and observing that
\[ \|T_{j}^{-1}(E + i\epsilon)  \phi_j^+\| \leq \|T_{j}^{-1}(E + i\epsilon) \phi_j^-\| + \|T_{j}^{-1}(E + i\epsilon) \| \|\phi_j^+ - \phi_j^-\|~,\] 
we deduce that
\begin{equation}\label{eq:Udet}
|\det U_j| \geq \|\phi_j^+ - \phi_j^-\| \geq \frac{e^{\gamma_j} - e^{-\gamma_j}}{\| T_j^{-1}(E+i\epsilon)\| }\geq\frac{\sqrt{\delta}}{2C^q}~.\end{equation}
Let $\psi_j^\pm$ be the unit eigenvectors of $(T_j^{-1}(E + i\epsilon))^*$, then using (\ref{eq:normT})
\[ 
|\langle\psi_{j+1}^\mp,  \phi_j^\pm  \rangle|= \frac{\| (\psi_{j+1}^\mp)^* (T_{j+1}^{-1}(E + i\epsilon) - e^{\pm \gamma_j}) \phi_j^\pm \|}{|e^{\mp \gamma_{j+1}} - e^{\pm\gamma_{j}}|} \leq \frac{2 (2C)^q \eta}{\sqrt{\delta}}~.\]
Decompose
\[ \phi_{j}^\pm = a_j^{\pm} \phi_{j+1}^\pm + b_j^{\pm} \phi_{j+1}^\mp~;\]
then we can assume, possibly modifying the phases of $\phi_j^\pm$ (starting from the last
one and going backwards), that $a_j^\pm \geq 0$. Then, using (\ref{eq:Udet}),
\[ |b_j^{\pm}| = \frac{|\langle \phi_j^\pm, \psi_{j+1}^\mp\rangle|}{|\langle \phi_{j+1}^\mp, \psi_{j+1}^\mp\rangle|}
\leq \frac{4(2C^2)^q\eta}{\delta} \leq \frac{C_1^q \eta}{\delta}~, \,\,\,
|a_j^\pm - 1| \leq  \frac{C_1^q \eta}{\delta}~.\]
If $\eta = \delta^{2}/\tilde{C}^q$ with sufficiently large $\tilde{C}$ and using (\ref{eq:Udet}), we deduce that
\[ \| \phi_j^\pm - \phi_{j+1}^\pm \| \leq \frac{2 C_1^q \eta}{\delta} \leq \frac{1}{2} \gamma_j e^{-\gamma_j}|\det U_{j+1}|~.\]
Finally, $\phi_{\ell_0 - 1}^+$ is proportional to the vector 
\[u =  \left(\begin{array}{cc}
\widetilde{G}^{[1, \infty)}(1, l_{0}q + 1)\\
\widetilde{G}^{[1, \infty)}(1, l_{0}q)\\
\end{array}\right)~.\]
Applying Theorem~\ref{matrixTheorem} with $\beta = 1/2$  we obtain:
\[
\frac{1}{\|u\|} \|\Phi^{-1}_{l_0q}\left(u\right)\| \geq \frac{1}{2} e^{\frac{1}{2} \sum_{i =
1}^{l_0 - 1}{\gamma_i}} \geq \frac{1}{2} e^{\frac{c}{2} l_0\sqrt{\delta}}~,
\]
whence, recalling (\ref{matrix_eq_for_tildeG}), the  rough bound
(\ref{eqboundforgreenfunctiononremainingintervals}), and
the definition $l_0 =
[\frac{\widetilde{c}\widetilde{q}\sqrt{\delta}}{q}]$,
\[
|\widetilde{G}^{[1, \infty)}(1, l_{0}q)| \leq \|u\|
\leq\frac{4}{\epsilon} e^{- \frac{c}{2} l_0\sqrt{\delta}} \leq
\frac{4}{\epsilon} e^{- \frac{c'\delta \tilde{q}}{q}}~. \]
This concludes the proof of (ii).

\medskip
The proof of the version of Theorem~\ref{lebesgueMeasure} with $J_\delta^{(2)}$
is essentially the same as the one given above for $J_\delta^{(1)}$.
The only difference  appears in the choice of $l_0$, as follows.
Let 
\[ S_-\left(\frac{p}{q}, 2\right) = \{E_1, \cdots, E_q \}~, \quad \text{where} \,\, \Delta_{p/q,2}(E_i) = 0
\,\, \text{for} \,\, i = 1,2,\cdots, q~. \]
Then the estimate
$|\Delta_{p/q,2}(E)|>\delta$ in the definition of $J_\delta$ is replaced
with  $\min_{1 \leq i \leq q} |E -
E_i| > \delta$. From the inequality
(\ref{eq.dot}) below, the latter condition implies the
former one with $\frac{\delta q}{e}$ instead of $\delta$.
\qed

\section{Proof of Theorem $\ref{mainTheorem}$}\label{sec:mainTheorem}

First, we prove a sufficient condition for $\dimension\,( S(\alpha, 2)) = 0$. Throughout
this section, we assume that all the fractions are reduced, i.e., of the form $\frac{p}{q}$,
where $q \geq 0$ and $p,q$ are relatively prime. Also, we follow the notation of
Theorem~\ref{lebesgueMeasure}, with $J_\delta = J_\delta^{(1)}$.

\begin{lemma}\label{lemma:suff}
Let $\alpha \in \R \setminus \Q$ be such that there exists a sequence
of rational approximations
\[ \alpha_j = \frac{p_j}{q_j} \underset{{j \to \infty}}{\longrightarrow} \alpha \]
for which the following conditions hold true for $j = 1,3,5,7,\dots$ :
\begin{enumerate}
\item $|\alpha_{j+1} - \alpha_j| < \min(\eta(p_j, q_j, q_j^{-j}), q_j^{-j})$.
 \item $|\alpha - \alpha_{j+1}| < q_{j+1}^{-(j+1)}$.
\item $q_{j+1} > q_j^j$ and also
$\frac{C(p_j,q_j)}{q_j^{-j}} \exp \left\{-\frac{q_j^{-j} q_{j+1}}{C(p_j,q_j)}\right\}  \leq \frac{1}{q_{j+1}^j}$,
\end{enumerate}
Then $\dimension\, (S(\alpha, 2)) = 0$.
\end{lemma}

\begin{proof}
Fix an odd $j \geq 1$. From Proposition~\ref{continuity_of_spectra},
\begin{equation}\label{eq:cont}
S(\alpha, 2) \subset S(\alpha_{j+1}, 2)
    + \left(- 6 \sqrt{2 |\alpha - \alpha_{j+1}|}, \, 6 \sqrt{2 |\alpha - \alpha_{j+1}|}\right)~,
\end{equation}
where for $A, B \subset \R$
\[ A + B = \left\{ a + b \, \big| \, a \in A, \, b \in B \right\}~. \]
Now we consider the operator $H_{\alpha_j, 2, \theta}$ corresponding to the fixed $\alpha_j$,
and let $\Delta(E) = \Delta_{\alpha_j, 2}(E)$ be as in (\ref{chamber_formula}). Recall that
\[ J_\delta = \left\{ E \, \big| \, |\Delta(E)| > \delta \right\}~,\]
and we set $\delta = q_j^{-j}$. Then
\[ S(\alpha_{j+1}, 2) = \Big(S(\alpha_{j+1}, 2) \cap J_\delta \Big)
    \bigcup \Big( S(\alpha_{j+1}, 2) \cap J_\delta^c \Big)~, \]
where $J_\delta^c = \R \setminus J_\delta$. Therefore
\begin{multline}\label{eq:cont2}
S(\alpha, 2) \subset\\
    \left[ \Big( S(\alpha_{j+1}, 2) \cap J_\delta \Big)
        + \left(- 6 \sqrt{2 |\alpha - \alpha_{j+1}|}, \, 6 \sqrt{2 |\alpha - \alpha_{j+1}|}\right) \right] \\
    \bigcup
    \left[ \Big( S(\alpha_{j+1}, 2) \cap J_\delta^c \Big)
        + \left(- 6 \sqrt{2 |\alpha - \alpha_{j+1}|}, \, 6 \sqrt{2 |\alpha - \alpha_{j+1}|}\right)\right]~.
\end{multline}
The set $S(\alpha_{j+1}, 2) \cap J_\delta$ is the union of at most $q_{j+1}$ intervals.
By Theorem~\ref{lebesgueMeasure} and assumption (3)
\[ \meas( S(\alpha_{j+1}, 2) \cap J_\delta)
    \leq \frac{C(p_j,q_j)}{q_j^{-j}} \exp \left\{-\frac{q_j^{-j} q_{j+1}}{C(p_j,q_j)}\right\}  \leq \frac{1}{q_{j+1}^j}~.\]
Therefore the set
\[ \Big( S(\alpha_{j+1}, 2) \cap J_\delta \Big)
        + \left(- 6 \sqrt{2 |\alpha - \alpha_{j+1}|}, \, 6 \sqrt{2 |\alpha - \alpha_{j+1}|}\right) \]
can be covered by $q_{j+1}$ intervals of total length smaller than
\[ \frac{1}{q_{j+1}^j} + 2 q_{j+1} \cdot 6 \sqrt{2 |\alpha - \alpha_{j+1}|}
    \leq \frac{1}{q_{j+1}^j} + \frac{C}{q_{j+1}^{\frac{j-1}{2}}} \leq \frac{C'}{q_{j+1}^{\frac{j-1}{2}}}~,\]
where $C,C'>0$ are some constants. Here the first inequality follows from assumption (2).


Now we pass to the second set in (\ref{eq:cont2}). We have:
\[ S(\alpha_{j+1}, 2) \cap J_\delta^c \subset J_\delta^c~, \]
and $J_\delta^c$ is the union of $q_j$ intervals:
\[ J_\delta^c = \bigcup_{\nu = 1}^{q_j} I_\nu
= \bigcup_{\nu=1}^{q_j} \left[ E^1_\nu, \, E^2_\nu \right]~. \]
Let us show that
\begin{equation}\label{eq:jdeltac}
\meas(J_\delta^c) \leq \frac{2e \delta}{q_j}~.
\end{equation}

Consider a band of $J_\delta^c$:\ $I_\nu =
[E^1_\nu,E^2_\nu]$. Denote by $E_\nu$\ the zero of $\Delta(E)$\ inside
$I_\nu$.\ Then
\[
\meas(I_\nu) = |E^1_\nu - E_\nu| + |E^2_\nu - E_\nu| \leq 2 |\widetilde{E}_\nu - E_\nu|~,\]
where $\widetilde{E}_\nu \in I_\nu$ is either $E^1_\nu$ or $E^2_\nu$;
$|\Delta(\widetilde{E}_\nu)| = \delta$.

Let us show that for $E \in I_\nu$
\begin{equation}\label{eq.dot}
|E - E_\nu | \leq \frac{e |\Delta(E)|}{|\Delta'(E)|}~.
\end{equation}
Then,
\[ \meas(I_\nu) \leq \frac{2 e \delta}{|\Delta'(E_\nu)|}~, \]
and by Lemma~\ref{last_lemma} we obtain
\[ \meas(J_\delta^c) \leq \sum_{k = 1}^{q_j}\frac{2e|\Delta(E)|}{|\Delta^{\prime}(E_\nu)|}
    \leq \frac{2e\delta}{q_j}~,\]
which proves (\ref{eq:jdeltac}).

\vspace{3mm}\noindent
We can assume without loss of generality that $\Delta(\cdot)$\ is increasing on
$I_\nu$ and that $\Delta(E) > 0$. We also assume for now that $\nu \neq 1, q_j$.

Denote by $E^0_\nu$\ the maximum of $\Delta(E)$\ just above
$E^2_\nu$.\  We have: $E^0_\nu > E^2_\nu$, since
$\Delta(E_\nu^2)= \delta < 2 \leq \Delta(E_\nu^0)$. Suppose $E^0_\nu > E > E_\nu$.\ Define
\begin{equation}\label{eq21}
f(E)\equiv {\frac{d}{dE}\ln{(\Delta(E))}} =
\frac{\Delta^\prime(E)}{\Delta(E)} = \sum_{k = 1}^{q_j}{\frac{1}{(E - E_k)}},
\end{equation}
and observe that
\begin{equation}\label{eq22}
f'(E) = \frac{d}{dE}f(E) = -\sum_{k = 1}^{q_j}{\frac{1}{(E -
E_k)^2}} < \frac{-1}{(E - E_\nu)^2}.
\end{equation}
Since $E^0_\nu$\ is the maximum of $\Delta(E)$,\ we have: $f(E^0_\nu)= 0$,
and for all $E\in{(E_\nu,E^0_\nu)}$,
\begin{multline}\label{eq23}
f(E) = -\int_{E}^{E^0_\nu}{f'(E')dE'} \\
> \int_{E}^{E^0_\nu}{\frac{dE'}{(E' - E_\nu)^2}} = \frac{1}{(E -
E_\nu)} - \frac{1}{(E^0_\nu - E_\nu)}.
\end{multline}
Now, consider $E\in{(E_\nu,E^2_\nu)}$,\ $E\in J_{\delta}^c$.\ We
have
\begin{multline}\label{limit_estim_for_E-E_nu}
\ln{\frac{\Delta(E)}{\Delta{(\widetilde{E})}}} =
\int_{\widetilde{E}}^{E}{f(E')dE'}
>
\ln\left(\frac{E -
E_\nu}{\widetilde{E} - E_\nu}\right) - \frac{E - \widetilde{E}}{E^0_\nu -
E_\nu},
\end{multline}
therefore, we obtain
\[
\frac{\Delta(E)}{\Delta(\widetilde{E})}
> \frac{1}{e}\left(\frac{E -
E_\nu}{\widetilde{E} - E_\nu}\right),
\]
which implies
\begin{equation}\label{estim_length_E-E_nu}
E - E_\nu < e\Delta(E)\lim_{\widetilde{E}\rightarrow
E_\nu}\frac{\widetilde{E} - E_\nu}{\Delta(\widetilde{E})} =
\frac{e\Delta(E)}{\Delta^\prime(E_\nu)}.
\end{equation}
This concludes the proof of (\ref{eq.dot}) for non-extremal bands. For extremal
bands we follow the strategy of \cite{last3}: due to monotonicity of $\Delta'$,
it is sufficient to consider the part of
the band directed towards the rest of the spectrum, and for this part of the band 
the previous argument is applied. Thus (\ref{eq:jdeltac}) is established. 

\vspace{2mm}
We have shown that $S(\alpha_{j+1}, 2) \cap J_\delta^c$ can be covered by $q_j$
intervals of total length at most $\frac{2e \delta}{q_j}$. Therefore
\[ \left[ \Big( S(\alpha_{j+1}, 2) \cap J_\delta^c \Big)
        + \left(- 6 \sqrt{2 |\alpha - \alpha_{j+1}|}, \, 6 \sqrt{2 |\alpha - \alpha_{j+1}|}\right)\right] \]
can be covered by $q_j$ intervals of total length at most
\[ \frac{2e \delta}{q_j} + 2 q_j \cdot 6 \sqrt{2|\alpha - \alpha_{j+1}|}
    \leq \frac{2e}{q_j^{j+1}} + \frac{Cq_j}{q_{j+1}^\frac{j+1}{2}} \leq \frac{C''}{q_j^{j-1}}~,\]
where the first inequality follows from assumption (2) and the definition of $\delta = q_j^{-j}$, and the
second inequality follows from assumption (3) of the Lemma.

Thus $S(\alpha, 2)$ can be covered by the union of $q_{j+1}$ intervals of total length
at most $\frac{C'}{q_{j+1}^{\frac{j-1}{2}}}$ and $q_{j}$ intervals of total length at most
$\frac{C''}{q_j^{j-1}}$. For any $\beta > 0$, the numbers $\frac{j-1}{2}$, $j-1$ exceed $\beta$
for sufficiently large $j$, therefore Lemma~\ref{hausdorff_dim_lemma} yields
\[ \dimension\, ( S(\alpha, 2)) \leq \inf_{\beta > 0} \frac{1}{1 + \beta} = 0~. \]
\end{proof}

Irrational $\alpha \in [0, 1]$ satisfying the assumptions of Lemma~\ref{lemma:suff} can
be constructed via their continued fraction expansion (see, e.g., \cite{hardy_wright}):
\begin{equation}\label{alpha}
\alpha = [n_1, n_2, n_3, \ldots]
    = \frac{1}{n_1 + \frac{1}{n_2 + \frac{1}{n_3 + \ldots}}}~.
\end{equation}
For example, the assumptions of Lemma~\ref{lemma:suff} are satisfied if the sequence of quotients,
$\{ n_j \}_{j=1}^\infty$, grows sufficiently fast. More generally, it is sufficient to control a subsequence of quotients $\{n_{j_k}\}_{k=1}^\infty$
so that the assumptions of Lemma~\ref{lemma:suff} hold, to ensure that, for the resulting
$\alpha$, we have $\dimension\, (S(\alpha, 2)) = 0$.

The set of $\alpha$'s obtained in this way can be shown to be dense $G_\delta$. However, it is
technically simpler to show directly that the set of $\alpha$'s satisfying the assumptions of
Lemma~\ref{lemma:suff} contains a dense $G_\delta$ set. The following lemma, combined with
Lemma~\ref{lemma:suff}, implies Theorem~\ref{mainTheorem}.

\begin{lemma}\label{lemma:gdelta} There exists a dense $G_\delta$ set of $\alpha$'s in $\R$ which satisfy the
assumptions of Lemma~\ref{lemma:suff}.
\end{lemma}

\begin{rmk*} As  mentioned, one can consider  only $\alpha \in [0, 1]$; however,
here we work with arbitrary $\alpha$ to keep the notation simpler.
\end{rmk*}

\begin{proof}
We construct the $G_\delta$ set $\mathcal{K}$ as an intersection
of open sets. Set $\mathcal{K}_0 = \R$. For $j = 0,1,2,\ldots$, we define $\mathcal{K}_{j+1}$ inductively
as follows. For $\frac{p}{q} \in \mathcal{K}_{j}$, one can choose $\epsilon = \epsilon(\frac{p}{q}, j)$ sufficiently
small so that 
\[ \epsilon < \min( \eta(p,q, q^{-j}), q^{-j})~, \,\,
\left(\frac{p}{q} - \epsilon, \frac{p}{q}+\epsilon\right) \subset \mathcal{K}_j~,\]
and the following conditions hold for any
\[\frac{\widetilde{p}}{\widetilde{q}} \in \left( \frac{p}{q}-\epsilon, \, \frac{p}{q}\right)
    \cup \left(\frac{p}{q}, \, \frac{p}{q} + \epsilon \right)~:\]
\begin{enumerate}
\item $\widetilde{q} > q^j$;
\item $\frac{C(p,q)}{q^{-j}} \exp\left\{- \frac{q^{-j}\widetilde{q}}{C(p, q)} \right\} \leq \frac{1}{\widetilde{q}^j}$. 
\end{enumerate}
Indeed, there is only a finite number of $\frac{\widetilde{p}}{\widetilde{q}} \in (\frac{p}{q}-1, \frac{p}{q}+1)$
which fail   either of the conditions (1)--(2); therefore, for sufficiently small $\epsilon > 0$, none of them
is in the punctured $\epsilon$-neighborhood of $\frac{p}{q}$.

Then denote
\[ \mathcal{K}_{j+1} = \bigcup_{\frac{p}{q} \in \mathcal{K}_j} \left( \frac{p}{q}-\epsilon(\frac{p}{q}, j), \, \frac{p}{q}\right)
    \cup \left(\frac{p}{q}, \, \frac{p}{q} + \epsilon(\frac{p}{q}, j) \right)~. \]
According to the choice of $\epsilon$, $\mathcal{K}_0 \supset \mathcal{K}_1 \supset \mathcal{K}_2 \supset \cdots$,
and all these sets are open. Set $\mathcal{K} = \bigcap_{j=1}^\infty \mathcal{K}_j$; it is a $G_\delta$
set.

Next, the set $\mathcal{K}$ is dense. Indeed, for any $\alpha \in \R$ and any $\eta > 0$, there
exists $\frac{p_0}{q_0} \in \R$ so that $|\alpha - \frac{p_0}{q_0}| < \frac{\eta}{2}$.
Choose a non-degenerate closed interval $I_0$ so that
\[ I_0 \subset \left(\frac{p_0}{q_0}, \frac{p_0}{q_0} + \min\left\{\epsilon\left(\frac{p_0}{q_0}, 0\right), \frac{\eta}{2}\right\}\right) \subset \mathcal{K}_1~. \]
Now proceed inductively: for any $j\geq 1$, we choose a rational $\frac{p_j}{q_j}$ in the interior of $I_{j-1}$
and choose $\eta_j > 0$ so that
$(\frac{p_j}{q_j}, \frac{p_j}{q_j} + \eta_j) \subset I_{j-1} \cap \mathcal{K}_{j+1}$, and then choose a non-degenerate
closed interval $I_j \subset (\frac{p_j}{q_j}, \frac{p_j}{q_j} + \eta_j)$. This is possible
since $\frac{p_j}{q_j}$ is an interior point of $I_{j-1}$ (therefore $(\frac{p_j}{q_j}, \frac{p_j}{q_j} + \eta_j) \subset I_{j-1}$ for sufficiently
small $\eta_j$), and is also a rational point in $I_{j-1} \subset \mathcal{K}_j$ (therefore $(\frac{p_j}{q_j}, \frac{p_j}{q_j} + \eta_j) \subset \mathcal{K}_{j+1}$
for $\eta_j < \epsilon(\frac{p_j}{q_j}, j)$). According to Cantor's lemma, the nested sequence of
closed intervals $\{I_j\}_{j=0}^\infty$ has a common point $\alpha' \in \mathcal{K}$, and $|\alpha - \alpha'| < \eta$.

Finally, let us show that all the elements of $\mathcal{K}$ satisfy the assumptions of Lemma~\ref{lemma:suff}.
Choose $\alpha \in \mathcal{K}$. First, $\alpha$ is irrational. Indeed, it is sufficient to show that
$\mathcal{K}_j$ does not contain rational numbers $\frac{\widetilde{p}}{\widetilde{q}}$  with
$\widetilde{q} \leq j$. For $j = 0$, the statement is empty. Assume that, for any
$\frac{p}{q} \in \mathcal{K}_j$, $q > j$. Then, by condition (1) above, for any
$\frac{\widetilde{p}}{\widetilde{q}} \in \mathcal{K}_{j+1}$ we have:
\[ \widetilde{q} > (j+1)^j \geq j+1~. \]
This proves the induction step.

Next, for any $j = 1,3,5,7,\dots$, we have: $\alpha \in \mathcal{K}_{j+2}$; therefore, there exists
$\frac{p_{j+1}}{q_{j+1}} \in \mathcal{K}_{j+1}$ such that
\[ \left|\alpha - \frac{p_{j+1}}{q_{j+1}}\right| < \epsilon\left(\frac{p_{j+1}}{q_{j+1}}, j+1\right)~. \]
Now, there exists $\frac{p_j}{q_j} \in \mathcal{K}_{j}$ so that
\[ 0 < \left|\frac{p_{j+1}}{q_{j+1}} - \frac{p_j}{q_j}\right| < \epsilon\left(\frac{p_{j}}{q_{j}}, j\right)\]
Thus we have constructed the  sequence $\frac{p_j}{q_j}$ required in Lemma~\ref{lemma:suff}.

\end{proof}


\begin{thebibliography}{}
\bibitem{AJ} A.~Avila, S.~Jitomirskaya, \textit{The Ten Martini Problem}, Ann.\ of Math.\ (2) 170 (2009), no.\ 1, 303--342.

\bibitem{simon_vanmouche_avron} J.~Avron,\ P.~van~ Mouche and
B.~Simon, \textit{On the measure of the spectrum for the almost
Mathieu operator}. Commun.\ Math.\ Phys.\ $\mathbf{132}$,\ 103--118
(1990)



\bibitem{conjhalf2} J.~Bell and R.~B.~Stinchcombe, \textit{Hierarchical band clustering and
fractal spectra in incommensurate systems}. J.\ Phys.\ \textbf{A}
$\mathbf{20}$,\ L739--L744 (1987)




\bibitem{Bourg} J.\ Bourgain, 
Green's function estimates for lattice Schr�dinger operators and applications. Annals of Mathematics Studies, 158. Princeton University Press, Princeton, NJ, 2005. x+173 pp. 

\bibitem{chembers} W.~Chambers, \textit{Linear network model for magnetic breakdown in
two dimensions}. Phys.\ Rev.\ \textbf{A} $\mathbf{140}$,\ 135--143
(1965)

\bibitem{choi-elliot-yui} M.~D.~Choi, G.~A.~Elliot and N.~Yui,
\textit{Gauss polynomials and the rotation algebra}, Invent.\ Math.\
$\mathbf{99}$,\ 225--246 (1990)






\bibitem{Falconer} K.~J.~Falconer,The geometry of fractal
sets. Cambridge: Cambridge University Press 1985



\bibitem{conjhalf3} T.~Geisel, R.~Ketzmerick and G.~Petshel, \textit{New class of
level statistics in quantum systems with unbounded diffusion}.
Phys.\ Rev.\ Lett.\ $\mathbf{66}$,\ 1651--1654 (1991)

%


\bibitem{hardy_wright} G.~H.~Hardy, E.~M.~Wright, 
An introduction to the theory of numbers. Fifth ed. Oxford: Oxford University Press 1979


\bibitem{HS} B.~Helffer and J.~Sj\"{o}strand,
\textit{Analyse semi-classique pour l'\'equation de Harper (avec application \`a l'\'equation de Schr�dinger avec champ magn\'etique)},
 M\'em. Soc. Math. France (N.S.) No. 34 (1988), 113 pp. (1989). 
  
\bibitem{Hofstadter} D.~R.~Hofstadter, \textit{Energy levels and wave functions of
Bloch electrons in a rational or irrational magnetic field}. Phys.\
Rev.\ B.\ $\mathbf{14}$,\ 2239--2249 (1976)

\bibitem{J}  S.\ Ya.\ Jitomirskaya, 
\textit{Almost everything about the almost Mathieu operator. II},
 XIth International Congress of Mathematical Physics (Paris, 1994), 373--382, Int. Press, Cambridge, MA, 1995. 





\bibitem{last2} Y.~Last, \textit{On the measure of gaps and spectra for
discrete 1D Schr\"{o}dinger operators}. Commun.\ Math.\ Phys.\
$\mathbf{149}$,\ 347--360 (1992)

\bibitem{last3} Y.~Last, \textit{Zero measure spectrum for the almost
Mathieu Operator}. Commun.\ Math.\ Phys.\ $\mathbf{164}$,\ 421--432
(1994)

\bibitem{last4} Y.\ Last, \textit{Almost everything about the almost Mathieu operator. I.}
XIth International Congress of Mathematical Physics (Paris, 1994), 366--372, 
Int. Press, Cambridge, MA, 1995

\bibitem{last_review} Y.~Last, \textit{Spectral theory of Sturm-Liouville
Operators on infinite intervals: A review of recent developments},
in W.~O~Amrein, A.~M.~Hinz and D.~B.~Pearson Eds.,
\textit{Sturm-Liouville Theory: Past and Present}, pp. 99--120,
Basel/Switzerland: Birkh\"{a}user Verlag 2005


\bibitem{last_wilkinson} Y.~Last and M.~Wilkinson, \textit{A sum rule for the dispersion
relations of the rational Harper's equation}. J.~Phys. \textbf{A}
$\mathbf{25}$,\ 6123--6133 (1992)


\bibitem{vmch} P.M.H.~van~Mouche, \textit{The coexistence problem for the discrete Mathieu operator}.
Commun.\ Math.\ Phys. $\mathbf{122}$,\ 23--34 (1989)



%

\bibitem{surace} S.~Surace, \textit{Positive Lyapunov exponent for a class of
ergodic Schr\"{o}dinger operators}. Commun.\ Math.\ Phys.\
$\mathbf{162}$,\ 529--537 (1994)

\bibitem{conjhalf1} C.~Tang and M.~Kohmoto, \textit{Global scaling properties of the
spectrum for a quasiperiodic Schr\"{o}dinger equation}. Phys.\ Rev.\
\textbf{B} $\mathbf{34}$,\ 2041--2044 (1986)



\bibitem{wilkonson_austin} M.~Wilkinson and E.~J.~Austin, \textit{Spectral
dimension and dynamics for Harper's equation}. Phys.\ Rev.\
\textbf{B} $\mathbf{50}$,\ 1420--1430 (1994)
 \end{thebibliography}
\end{document}